\documentclass[a4paper, 12pt]{article}

\usepackage[sort&compress]{natbib}
\bibpunct{(}{)}{;}{a}{}{,} 

\usepackage{amsthm, amsmath, amssymb, mathrsfs, multirow, url, subfigure}
\usepackage{graphicx} 
\usepackage{ifthen} 
\usepackage{amsfonts}
\usepackage[usenames]{color}
\usepackage{fullpage}

\theoremstyle{plain} 
\newtheorem{thm}{Theorem}

\theoremstyle{definition}

\theoremstyle{remark} 
\newtheorem{remark}{Remark}

\newcommand{\prob}{\mathsf{P}}

\newcommand{\bel}{\mathsf{bel}}
\newcommand{\pl}{\mathsf{pl}}

\newcommand{\bin}{{\sf Bin}}
\newcommand{\unif}{{\sf Unif}}
\newcommand{\nm}{{\sf N}}

\newcommand{\gam}{{\sf Gamma}}
\newcommand{\stt}{{\sf t}}

\newcommand{\chisq}{{\sf ChiSq}}
\newcommand{\be}{{\sf Beta}}

\newcommand{\B}{\mathcal{B}}

\newcommand{\YY}{\mathbb{Y}}
\newcommand{\UU}{\mathbb{U}}

\newcommand{\WW}{\mathbb{W}}

\renewcommand{\S}{\mathcal{S}}

\renewcommand{\phi}{\varphi}

\newcommand{\stgeq}{\geq_{\text{st}}}

\title{Prior-free probabilistic prediction of future observations}
\author{
Ryan Martin \\
Department of Mathematics, Statistics, and Computer Science \\
University of Illinois at Chicago \\
\url{rgmartin@uic.edu} \\
\mbox{} \\
Rama T.~Lingham \\
Division of Statistics \\
Northern Illinois University \\
\url{rama@math.niu.edu} 
}
\date{\today}

\begin{document}

\maketitle 

\begin{abstract}
Prediction of future observations is a fundamental problem in statistics.  Here we present a general approach based on the recently developed inferential model (IM) framework.  We employ an IM-based technique to marginalize out the unknown parameters, yielding prior-free probabilistic prediction of future observables.  Verifiable sufficient conditions are given for validity of our IM for prediction, and a variety of examples demonstrate the proposed method's performance.  Thanks to its generality and ease of implementation, we expect that our IM-based method for prediction will be a useful tool for practitioners.  

\smallskip

\end{abstract}

\section{Introduction}
\label{S:intro}

The prediction of future observations based on the information available in a given sample is a fundamental problem in statistics.  For example, in engineering applications, such as computer networking, one might want to predict, to some degree of certainty, the time at which the current system might fail in order to have resources available to fix it.  Despite this being a fundamental problem, the available literature does not seem to give any clear guidelines about how to approach a prediction problem in general.  From a frequentist point of view, there are a host of techniques available for constructing prediction intervals in specific examples; see, for example, the book by \citet{hahn.meeker.1991} and papers by \citet{fertig.mann.1977}, \citet{bhaumik.gibbons.2004}, \citet{krishnamoorthy.etal.2008}, and \citet{wang2010}.  Some general approaches to the frequentist prediction problem are presented in \citet{beran1990} and \citet{lawless.fredette.2005}.  From a Bayesian point of view, if a prior distribution is available for the unknown parameter, then the prediction problem is conceptually straightforward.  The Bayesian model admits a joint distribution for the observed data and future data, so a conditional distribution for the latter given the former---the Bayesian predictive distribution---is the natural tool.  \citet{hamada.etal.2004} presents some applications of Bayesian prediction.  The catch is that often there is no clear choice of prior.  Default, or non-informative priors can be used but, in that case, it is not clear that the resulting inference will be meaningful in either a personal probabilistic or frequentist sense.  \citet{wang.hannig.iyer.2012} propose a fiducial approach for prediction which, at a high-level, can be viewed as a sort of compromise to the frequentist and Bayesian approaches.  They propose a very natural predictive distribution that obtains from the usual fiducial distribution for the parameter \citep{hannig2009, hannig2012}.  They show  that the prediction intervals obtained from the fiducial predictive distribution are asymptotically correct and perform well, in examples, compared to existing prediction intervals.   

The fiducial approach is attractive because no prior distributions are required.  However, like Bayesian posterior or predictive distributions based on default priors, fiducial distributions may not be calibrated for meaningful probabilistic inference, except possibly in the limit \citep{liu.martin.wire}.  Recently, \citet{imbasics} proposed a general framework for prior-free probabilistic inference, called \emph{inferential models} (IMs).  This framework has some parallels with fiducial \citep[e.g.,][]{fisher1959, hannig2009, hannig2012} and Dempster--Shafer theory \citep[e.g.,][]{dempster2008, shafer1976} in that work is carried out in terms of unobservable auxiliary variables.  There are also some connections with the frequentist confidence distributions \citep[e.g.,][]{xie.singh.2012}.  The key difference between IMs and these other frameworks is the way that the auxiliary variables are handled.  The important feature is that IMs provide probabilistic summaries of the information in data concerning the quantity of interest, and that these summaries are provably meaningful, not just in an asymptotic sense.  

In this paper, we provide a IM-based solution to the problem of predicting future observations.  The critical observation that drives the approach here is that predicting future observations is a marginal inference problem, one in which the full parameter itself is a nuisance parameter to be marginalized out.  With this view, in Section~\ref{SS:setup} we apply the general marginalization principles in \citet{immarg} to eliminate the nuisance parameter, directly providing a marginal IM for the future observations.  In Section~\ref{SS:pred.valid} we give general conditions under which the resulting IM for prediction is valid, and discuss under what circumstances these conditions hold, in what cases they can be weakened, and other consequences.  The key point is that the plausibility function obtained from a valid IM provides a probabilistic summary of the information in the observed data concerning the future data to be predicted; this function can be plotted to provide some visual summary.  Moreover, the validity theorem demonstrates that the predictive plausibility interval, defined in \eqref{eq:pred.region}, has the nominal frequentist coverage for all finite samples, not just in the limit.  Our focus here is on the case of predicting a univariate future observable, but the multivariate case, discussed briefly in Section~\ref{SS:multi}, requires some additional considerations.  Several practical examples of prediction in the IM context are worked out in Section~\ref{S:examples}.  These examples involve a variety of common models, and prediction problems in quality control, environmental, system breakdown, and disease count applications are considered.  To compare our IM-based solution to other existing methods, we focus on frequentist performance of our prediction intervals.  In all the examples we consider, the IM intervals are competitive with the existing methods.  R code for these examples is provided at \url{www.math.uic.edu/~rgmartin}.  The take-away message is that the IM approach provides an easily implementable and general method for constructing meaningful prior-free probabilistic summaries of the information in observed data  for inference or prediction; that these summaries can be converted to frequentist procedures with fixed-$n$ performance guarantees and comparable efficiencies compared to existing methods is an added bonus.

\section{Inferential models for prediction}
\label{S:prediction}

\subsection{Basic notation and terminology}

The basic IM framework is introduced in \citet{imbasics}, and further developments are presented in \citet{imcond, immarg}.  Here we want to briefly introduce the necessary notation and terminology.  Suppose that the goal is inference on an unknown parameter $\theta$.  \citet{imbasics} present a three-step IM construction: association (A), prediction (P), and combination (C) steps.  The starting point is identifying an association that links the data $Y$ and the parameter $\theta$ to an unobservable auxiliary variable $U$.  Often, a naive association will have auxiliary variables of higher dimension than the parameter, so special conditioning and/or marginalization techniques are needed to satisfactorily reduce the dimension of $U$.  In any case, once this ``baseline'' association is specified, the A-step of the IM construction is completed.  The P-step, unique to the IM approach, introduces a random set for predicting $U$.  Finally, the C-step combines the association with the predictive random set in a natural way, yielding a pair of belief and plausibility functions for probabilistic inference on $\theta$.  The aforementioned papers give a number of examples of this approach, along with further explanation and theory.

\subsection{Preview for prediction}
\label{SS:preview}

Before getting into the general details about the prediction problem, we present a relatively simple example as a preview of our proposed IM approach.  Consider a homogeneous Poisson process $\{N(t): t \geq 0\}$ with rate $\theta > 0$.  The arrival times $T_0,T_1,T_2,\ldots$ are such that $T_0 \equiv 0$ and the inter-arrival times $T_i-T_{i-1}$, $i \geq 1$, are independent exponential random variables with rate $\theta$.  

If the sampling scheme is to wait for the $n$-th arrival, then the sufficient statistic for $\theta$ in this model is $Y=T_n$, the last arrival time.  Based on the arguments in \citet{imcond}, the baseline association for $\theta$ is 
\[ Y = (1/\theta) G_n^{-1}(U), \quad U \sim \unif(0,1), \]
where $G_n$ is the ${\sf Gamma}(n,1)$ distribution function.  If inference on $\theta$ was the goal, then this would complete the A-step.  However, suppose the goal is to predict $\tilde Y = T_{n+k}$, the time of the $(n+k)$-th arrival, for some fixed integer $k \geq 1$.  Then $\theta$ itself is a nuisance parameter, and the quantity of interest is $\tilde Y$.  From the baseline association above, we can easily solve for $\theta$ in terms of $(Y,U)$, i.e.,  
\[ \theta(Y,U) = G_n^{-1}(U)/Y. \]
Since $\tilde Y=T_{n+k}$, for given $Y=T_n$, equals $Y$ plus an independent gamma random variable with shape $k$ and rate $\theta$, following \citet{immarg}, we have a marginal association for $\tilde Y$ given by 
\[ \tilde Y = Y + \frac{1}{\theta(Y,U)} G_k^{-1}(\tilde U) = Y \Bigl( 1 + \frac{G_k^{-1}(\tilde U)}{G_n^{-1}(U)} \Bigr). \]
This completes the A-step for prediction.  If $R$ denotes the ratio in the far right-hand side above, then $R$ has a generalized gamma ratio distribution \citep{coelho.mexia.2007} with density function $f(r) \propto (1 + r)^{-(n+k)}$, $r > 0$.  If $F$ is the corresponding distribution function, then we may rewrite the marginal association as 
\[ \tilde Y = Y \{ 1 + F^{-1}(W) \}, \quad W \sim \unif(0,1). \]
Thus, we have successfully marginalized out the unknown parameter, directly associating the quantity to be predicted, $\tilde Y$, to the observed data, $Y$, and an auxiliary variable, $W$.  Then, the general IM principles \citep{imbasics, imcond, immarg} can be applied directly.  In particular, we apply the P- and C-steps to the association for $\tilde Y$, resulting in prior-free probabilistic prediction of the future arrival time.  The next two subsections will describe the proposed approach in more detail, and our examples in Section~\ref{S:examples} will demonstrate its generality, its quality performance, and its simplicity in applications.

\subsection{General setup and the A-step}
\label{SS:setup}

In the prediction problem, there is observed data $Y$ and future data $\tilde Y$ to be predicted; the two are linked together through a common parameter $\theta$.  Here we assume that $\tilde Y$ is a scalar, though it could be a function of several future observations; see Section~\ref{SS:multi} for discussion on the multivariate prediction problem.  Write the sampling model $\prob_{Y|\theta}$ for $Y$ in association form:
\begin{equation}
\label{eq:amodel}
Y=a(\theta,U), \quad U \sim \prob_U, 
\end{equation}
where $\prob_U$ is known and free of $\theta$.  We call this the ``baseline'' association, and it connects observable data $Y$ and unknown parameter $\theta$ to an unobservable auxiliary variable $U$.  Despite its simple form, the baseline association is quite general, i.e., it covers cases outside the structural models in \citet{fraser1968}; see Sections~\ref{SS:gamma}--\ref{SS:binomial}.  For example, for any iid model with a smooth distribution function $F_\theta$, take the $i$-th component of $a(\theta,U)$ to be $F_\theta^{-1}(U_i)$ for $U_i \sim \unif(0,1)$.  Intuitively, any model that can be simulated has a form \eqref{eq:amodel}.  

As our first step, assume that this baseline association can be re-expressed as
\[ T(Y) = b(\theta, \tau(U)) \quad \text{and} \quad H(Y) = \eta(U), \]
for functions $(T,H)$ and $(\tau,\eta)$ such that $y \mapsto (T(y),H(y))$ and $u \mapsto (\tau(u),\eta(u))$ are one-to-one.  A key feature of this decomposition is that a solution $\theta=\theta(y,v)$ of the equation $T(y)=b(\theta,v)$ is available for all $(y,v)$.  By conditioning on the observed value, $H(Y)$, of $\eta(U)$, this association can then be reduced as follows:
\[ T(Y) = b(\theta, V), \quad V \equiv \tau(U) \sim \prob_{\tau(U)|\eta(U)=H(Y)}. \]
\citet{imcond} show that such a decomposition exists in broad generality.  For simplicity, we assume here that $\tau(U)$ and $\eta(U)$ are independent, so the conditioning can be dropped, i.e., $\prob_{\tau(U)|\eta(U)=H(Y)} \equiv \prob_{\tau(U)}$.  This assumption holds for many problems, including those in Section~\ref{S:examples}.  Dependence in this context is only a technical complication, not conceptual, so we focus here on the simpler case of independent $\tau(U)$ and $\eta(U)$; the dependent case is discussed further in Section~\ref{SS:conditional}.  

For the observed data $Y$ and the future data $\tilde Y$, write a joint association:
\[ T(Y)=b(\theta,V) \quad \text{and} \quad \tilde Y = \tilde a(\theta, \tilde U), \]
where $(V,\tilde U) \sim \prob_{(V,\tilde U)}$.  When $Y$ and $\tilde Y$ are independent, $V$ and $\tilde U$ are likewise independent, but in time series problems, for example, the auxiliary variables will be correlated.  The use of ``$\tilde a$'' for the mapping instead of simply ``$a$'' is to cover the case where $Y$ and $\tilde Y$ are related through a common parameter $\theta$, but possibly have different distributions.  For example, $Y$ might be an iid normal sample, while $\tilde Y$ is the maximum of ten future normal samples; similarly, in a regression context, $Y$ and $\tilde Y$ might have different values of the predictor variables.  

Solving for $\theta$ in the first equation and plugging in to the second gives 
\[ T(Y)=b(\theta,V) \quad \text{and} \quad \tilde Y = \tilde a\bigl(\theta(Y,V), \tilde U\bigr), \]
Since prediction is a marginal inference problem, where $\theta$ itself is the nuisance parameter, it follows from the general theory in \citet{immarg} that the first equation in the above display can be ignored.  This leaves a marginal association for $\tilde Y$:
\begin{equation}
\label{eq:pred.marg0}
\tilde Y = \tilde a\bigl( \theta(Y,V), \tilde U \bigr). 
\end{equation}
This marginalization has some similarities to the Bayesian and fiducial predictive distributions.  That is, the model for $\tilde Y$ in \eqref{eq:pred.marg0} is that of a mixture of the distribution of $\theta(Y,V)$, for fixed $Y$, with the distribution of $\tilde a(\theta,\tilde U)$ for fixed $\theta$.  This, of course, is not the ``true'' distribution of $\tilde Y$ given $Y$; the idea is that the future observable $\tilde Y$ is being modeled as a $Y$-dependent function of $(V,\tilde U)$.  We claim that equation \eqref{eq:pred.marg0} describes a sort of predictive distribution of $\tilde Y$ for a given $Y$, similar to the frequentist predictive distributions in, e.g., \citet{lawless.fredette.2005}.  To see this better, let $G_Y$ be the distribution of the right-hand side of \eqref{eq:pred.marg0} as a function of $(V,\tilde U)$ for fixed $Y$.  Then, in the case this is an absolutely continuous distribution, \eqref{eq:pred.marg0} can be rewritten as 
\begin{equation}
\label{eq:pred.marg}
\tilde Y = G_Y^{-1}(W), \quad W \sim \unif(0,1), 
\end{equation}
so $G_Y$ plays the role of a predictive distribution for $\tilde Y$.  This completes the A-step in the construction of the IM for prediction.  That is, \eqref{eq:pred.marg} is the association that links the observable data $Y$, the unobservable auxiliary variable $(U,\tilde U)$, and the future data $\tilde Y$.  

Though \eqref{eq:pred.marg0} has some connection to Bayesian and fiducial prediction, it differs from a plug-in or parametric bootstrap prediction.  The difference is that the quantity $\theta(Y,V)$ plugged in is not fixed.  That is, we consider the distribution of $\tilde a(\theta(Y,V),\tilde U)$ as a function of $(V,\tilde U)$, not the distribution of $\tilde a(\hat\theta_Y, \tilde U)$, as a function of $\tilde U$, for fixed $\hat\theta_Y$.

\subsection{P- and C-steps}
\label{SS:pc.steps}

After the A-step in \eqref{eq:pred.marg}, the P-step requires specification of a suitable predictive random set $\S \sim \prob_\S$ for $W$.  A rigorous presentation on the theory of random sets is given in \citet{molchanov2005}, including a general definition.  For our purposes here, it suffices to define a random set by first specifying a probability space $(\WW, \B_\WW, \prob_W)$ and a map $S: \WW \to \B_\WW$ which is measurable in the sense that $\{w: S(w) \cap K \neq \varnothing\} \in \B_\WW$ for all compact $K \subseteq \WW$.  Then $\S=S(W)$, for $W \sim \prob_W$ is a random set, and its distribution $\prob_\S$ is the push-forward measure $\prob_W S^{-1}$.  \citet{imbasics} argue that the choice of predictive random set ought to depend on the assertion $A$ of interest.  There are three kinds of assertions about $\tilde Y$ that will be of interest here in the prediction problem: two one-sided assertions, and a singleton assertion.  Given a predictive random set and an assertion of interest, the C-step proceeds by combining the A- and P-step results.  \citet{imbasics} give a general explanation, but here this amounts to computing the \emph{plausibility of $A$}, i.e.,  
\[ \pl_Y(A) = \prob_\S\{G_Y^{-1}(\S) \cap A \neq \varnothing\}. \]
Next we discuss, in turn, the P- and C-steps for each of these kinds of assertions.   
\begin{itemize}
\item {\it Right-sided.} A right-sided assertion is of the form $A=\{\tilde Y > \tilde y\}$ for a fixed $\tilde y$.  For this assertion, by Theorem~4 in \citet{imbasics}, the optimal predictive random set is one-sided: $\S=[0,W]$ for $W \sim \unif(0,1)$.  In this case, the C-step gives the plausibility function 
\begin{equation}
\label{eq:pl.rtside}
 \pl_Y(A) = \prob_\S\{G_Y^{-1}(\S) \cap A \neq \varnothing\} = 1-G_Y(\tilde y)
\end{equation}
The plausibility function is a non-increasing function of $\tilde y$; see Figure~\ref{fig:bhaumik}(a) described in Section~\ref{SS:log.normal}.  Hence the prediction region \eqref{eq:pred.region} based on the plausibility function in \eqref{eq:pl.rtside} will be an upper prediction bound for $\tilde Y$. 
\item {\it Left-sided.} A left-sided assertion is of the form $A=\{\tilde Y \leq \tilde y\}$ for a fixed $\tilde y$.  Similar to the right-sided case, the optimal predictive random set is $\S=[W,1]$ for $W \sim \unif(0,1)$.  Then the C-step gives the plausibility function 
\begin{equation}
\label{eq:pl.ltside}
 \pl_Y(A) = \prob_\S\{G_Y^{-1}(\S) \cap A \neq \varnothing\} = G_Y(\tilde y)
\end{equation}
The plausibility function is an non-decreasing function of $\tilde y$; see Figure~\ref{fig:hamada}(a).  Hence the prediction region \eqref{eq:pred.region} based on the plausibility function in \eqref{eq:pl.ltside} will be a lower prediction bound for $\tilde Y$. 
\item {\it Singleton.} A singleton assertion is of the form $A=\{\tilde Y = \tilde y\}$ for a fixed $\tilde y$.  The optimal predictive random set worked out in \citet{imbasics} for this assertion is complicated, but a natural choice that is suitable in most cases (and optimal in some cases) is the ``default'' predictive random set $\S=\{w: |w-0.5| \leq |W-0.5|\}$, for $W \sim \unif(0,1)$.  Then the C-step gives the plausibility function 
\begin{equation}
\label{eq:pl.2side}
 \pl_Y(A) = \prob_\S\{G_Y^{-1}(\S) \cap A \neq \varnothing\} = 1-|2G_Y(\tilde y)-1|
\end{equation} 
The prediction region \eqref{eq:pred.region} based on the plausibility function in \eqref{eq:pl.2side} will be a two-sided prediction bound for $\tilde Y$.  
\end{itemize}

It is important to note that, although the general P- and C-steps may appear rather technical, implementation of the IM approach for prediction requires only that one be able to evaluate, either analytically or numerically, the distribution function $G_Y$.  Section~\ref{S:examples} gives several examples and applications to demonstrate that our IM-based plausibility intervals are good general tools for the prediction problem, and that such intervals are often better than what other methods provide.

\subsection{Prediction validity}
\label{SS:pred.valid}

Here we give the main distributional property of the plausibility function for prediction.  The key requirement is a mild condition on the predictive random set $\S$.  Following \citet{randset}, define the contour function $f_\S(w) = \prob_\S(\S \ni w)$.  Then the predictive random set $\S$ is \emph{valid} if 
\begin{equation}
\label{eq:prs.valid}
f_\S(W) \stgeq \unif(0,1) \quad \text{when $W \sim \prob_W$}, 
\end{equation}
where $\stgeq$ means ``stochastically no smaller than.'' \citet{imbasics} demonstrate that this is a very mild condition.  (Though not required for the theorem, they also recommend to consider only predictive random sets with nested support.  Those discussed in the previous section are all nested.)  The three assertions $A$ described in Section~\ref{SS:pc.steps} depend on a generic $\tilde y$.  Here we write $\pl_Y(\tilde y)$ for the plausibility function for such an assertion; the specific kind of assertion will be clear from the context.  

\begin{thm}
\label{thm:pred.valid}
For the marginal association \eqref{eq:pred.marg} for $\tilde Y$, let $\S \sim \prob_\S$ be a valid predictive random set for $W \sim \unif(0,1)$, i.e., \eqref{eq:prs.valid} holds, which is non-empty with $\prob_\S$-probability~1.  If $G_Y(\tilde Y) \sim \unif(0,1)$ for $(Y,\tilde Y) \sim \prob_{(Y,\tilde Y)|\theta}$ for all $\theta$, then  
\[ \sup_\theta \prob_{(Y,\tilde Y)|\theta}\{\pl_Y(\tilde Y) \leq \alpha\} \leq \alpha, \quad \forall \; \alpha \in (0,1). \]
This holds whether $\pl_Y(\tilde Y)$ is based on right-sided, left-sided, or singleton assertions.
\end{thm}

\begin{proof}
Since $\pl_y(\tilde y) = f_\S(G_y(\tilde y))$, the result follows from the assumed validity of $\S$ and the assumption that $G_Y(\tilde Y) \sim \unif(0,1)$ as a function of $(Y,\tilde Y)$.  
\end{proof}

The following sequence of remarks discusses the assumptions, interpretations, and various extensions of Theorem~\ref{thm:pred.valid}.  See, also, Section~\ref{S:further.details}.  

\begin{remark}
\label{re:interpretation}
\citet{imbasics} argue that validity gives the plausibility function a scale on which the numerical values can be interpreted.  For example, like in the familiar case of p-values, if the plausibility function is small, e.g., $\pl_y(\tilde y) < 0.05$, then, for the given $Y=y$, the value $\tilde y$ is not a plausible prediction; see, also, Remark~\ref{re:pred.region}.  
\end{remark}

\begin{remark}
\label{re:pred.region}
A consequence of Theorem~\ref{thm:pred.valid} is that the set 
\begin{equation}
\label{eq:pred.region}
\{\tilde y: \pl_y(\tilde y) > \alpha\}
\end{equation}
is a $100(1-\alpha)$\% prediction plausibility region, i.e., the probability that $\tilde Y$ falls inside the region \eqref{eq:pred.region} is at least $1-\alpha$ under the joint distribution of $(Y,\tilde Y)$ for any parameter value $\theta$.  Then, for the three kinds of assertions, namely, right, left, and singleton, discussed in Section~\ref{SS:pc.steps}, one gets $100(1-\alpha)$\% upper, lower, and two-sided prediction intervals, respectively.  Moreover, the region \eqref{eq:pred.region} has the following desirable interpretation: each point $\tilde y$ it contains is individually sufficiently plausible.  No frequentist, Bayes, or fiducial prediction interval assigns such a meaning to the individual elements it contains.
\end{remark}

\begin{remark}
\label{re:efficiency}
Suppose that $\S$ is such that $f_\S(V) \sim \unif(0,1)$ for $V \sim \unif(0,1)$.  Then $\pl_Y(\tilde Y) \sim \unif(0,1)$ as a function of $(Y,\tilde Y) \sim \prob_{(Y,\tilde Y)|\theta}$ for all $\theta$.  These conditions hold in many examples (see Section~\ref{S:examples}) and they imply that the plausibility region in \eqref{eq:pred.region} has exact prediction coverage, $1-\alpha$, not just conservative.  
\end{remark}

\begin{remark}
\label{re:separable}
A natural question is: under what conditions does  $G_Y(\tilde Y) \sim \unif(0,1)$ hold?  An important example is the case we shall call ``separable,''  where the effect of $Y$ on the right-hand side of \eqref{eq:pred.marg0} can be separated from the auxiliary variables, i.e., \eqref{eq:pred.marg0} can be rewritten as $p(Y,\tilde Y) = \phi(V, \tilde U)$ for some functions $p$ and $\phi$.  In the language of \citet{lawless.fredette.2005}, the quantity $p(Y,\tilde Y)$ is an exact pivot.  Many problems with a group transformation structure \citep[e.g.][]{eaton1989} are separable, and are covered by Theorem~\ref{thm:pred.valid}.  Some of the examples in Section~\ref{S:examples} are of this type, but the numerical results even for the non-separable models (see Sections~\ref{SS:gamma}--\ref{SS:binomial}) suggest that the validity result holds broadly.  Section~\ref{SS:asymptotics} has more discussion on the non-separable case.  
\end{remark}

\begin{remark}
\label{re:fiducial}
An advantage of the IM's handling of the auxiliary variables, revealed in the previous remarks, is that one has finite-sample control on the prediction coverage.  The fiducial approach to prediction, on the other hand, can only guarantee asymptotic control of frequentist prediction coverage \citep[][Theorem~1]{wang.hannig.iyer.2012}.    
\end{remark}


\begin{remark}
\label{re:stochastic.order}
The uniformity condition in Theorem~\ref{thm:pred.valid} can be relaxed to a stochastic ordering condition, but then the conclusion holds only for certain predictive random sets and certain assertions.  For example, suppose $G_Y(\tilde Y)$ is stochastically no smaller than $\unif(0,1)$.  Then the conclusion of Theorem~\ref{thm:pred.valid} holds for the one-sided predictive random set $\S=[0,W]$, $W \sim \unif(0,1)$.  In this case, by taking assertions $A=\{\tilde Y > \tilde y\}$, the lower plausibility bounds obtained via \eqref{eq:pred.region} have the nominal frequentist coverage probability as described in Remark~\ref{re:pred.region}.  Similar conclusions hold if $G_Y(\tilde Y)$ is stochastically no larger than $\unif(0,1)$, with obvious changes to the predictive random set and assertion. 
\end{remark}

\section{Examples and applications}
\label{S:examples}

\subsection{Normal models and a quality control application}
\label{SS:normal}

Let $Y=(Y_1,\ldots,Y_n)$ be an iid sample from a $\nm(\mu,\sigma^2)$ population, where $\theta=(\mu,\sigma)$ is unknown.  Our first goal is to predict the next independent observation $\tilde Y = Y_{n+1}$.  To start, consider the baseline association involving the original data
\[ Y_i = \mu + \sigma Z_i, \quad i=1,\ldots,n, \]
where $Z_1,\ldots,Z_n$ are iid $\nm(0,1)$.  Based on the arguments in \citet{imcond}, a conditional IM for $\theta=(\mu,\sigma)$ has association 
\[ \bar Y = \mu + \sigma n^{-1/2} U_1 \quad \text{and} \quad S = \sigma U_2, \]
where $\bar Y=n^{-1}\sum_{i=1}^n Y_i$ is the sample mean, $S^2 = (n-1)^{-1}\sum_{i=1}^n (Y_i-\bar Y)^2$ is the sample variance, $U_1 \sim \nm(0,1)$, and $(n-1)U_2^2 \sim \chisq(n-1)$, with $U_1$ and $U_2$ independent.  Then it is easy to see that 
\[ \theta(Y,U) = \bigl( \mu(Y,U), \sigma(Y,U) \bigr) = \Bigl(\bar Y - \frac{S}{n^{1/2}} \frac{U_1}{U_2}\,,\, \frac{S}{U_2} \Bigr). \] 
For the next observation $\tilde Y = Y_{n+1}$, the association is just like the baseline association above, i.e., $\tilde Y = \mu + \sigma \tilde U$, where $\tilde U$ is independent of $(U_1,U_2)$.  As discussed above, we can insert $\theta(Y,U)$ in place of $\theta$ in this association to get a marginal association for $\tilde Y$: 
\begin{equation}
\label{eq:normal.association}
\tilde Y = \bar Y - \frac{S}{n^{1/2}}\frac{U_1}{U_2} + \frac{S}{U_2} \tilde U = \bar Y + S \Bigl( \frac{1}{n^{1/2}} \frac{U_1}{U_2} - \frac{\tilde U}{U_2} \Bigr). 
\end{equation}
This is clearly one of those separable cases as described in Remark~\ref{re:separable}.  Also,
\[ V = \frac{1}{n^{1/2}} \frac{U_1}{U_2} - \frac{\tilde U}{U_2} \]
is distributed as $(n^{-1}+1)^{1/2} \stt(n-1)$, with distribution function $F_n$.  Then the marginal association \eqref{eq:normal.association} can be written as $\tilde Y = \bar Y + S F_n^{-1}(W)$, with $W \sim \unif(0,1)$.  If we are interested in a two-sided prediction interval, then, as in Section~\ref{SS:pc.steps}, we take a singleton assertion $A=\{\tilde Y = \tilde y\}$ and get the following plausibility function:
\[ \pl_Y(\tilde y) = 1-\Bigl|2F_n\Bigl( \frac{\tilde y - \bar Y}{S} \Bigr) - 1 \Bigr|. \]
Then the corresponding two-sided $100(1-\alpha)$\% plausibility interval \eqref{eq:pred.region} for $\tilde Y$ is 
\[ \bar Y \pm t_{n-1, 1-\alpha/2}^\star S (1 + n^{-1})^{1/2}, \]
where $t_{\nu,p}^\star$ is the $100p$th percentile of the t-distribution with $\nu$ degrees of freedom.  This is exactly the classical Student-t prediction interval discussed in, e.g.,  \citet{geisser1993}.  

The ideas just discussed extend quite naturally to the case of normal linear regression.  The details of the IM calculations would be similar to those presented in \citet{wang.hannig.iyer.2012} for the fiducial case and, hence, omitted here.

As a more sophisticated example, \citet{odeh1990} gives a quality control application involving sprinkler systems for fire prevention in a hotel.  In this application, based on a sample of $n=20$ sprinklers, whose activation temperatures are normally distributed, the goal is to give a two-sided prediction interval for the temperature at which at least $k=36$ of $m=40$ new sprinklers will activate.  In other words, the goal is to predict the temperature at which at least $k$ of the $m$ new sprinklers will activate.  The IM methodology can be used for this problem.  Let $\tilde Y$ be the $k$-th largest of $m$ future independent normal observations $Y_{n+1},\ldots,Y_{n+m}$.  The corresponding association for $\tilde Y$ is 
\[ \tilde Y = \mu + \sigma \tilde U, \quad \text{where} \quad \tilde U = \text{$k$-th largest of $U_{n+1},\ldots,U_{n+m}$}, \]
and $U_{n+1},\ldots,U_{n+m}$ are iid $\nm(0,1)$.  Then the marginal association for $\tilde Y$ can be written exactly as in \eqref{eq:normal.association} and the problem is still separable.  The only difference here is that $V = (n^{-1/2} U_1 - \tilde U)/U_2$ has a non-standard distribution.  As before, write $\tilde Y = \bar Y + S F_{n,m,k}^{-1}(W)$, where $F_{n,m,k}$ is the distribution function of $V$, and $W \sim \unif(0,1)$.  The distribution $F_{n,m,k}$ can be simulated and, therefore, one can easily get a Monte Carlo approximation of the plausibility function \eqref{eq:pl.2side} for $\tilde Y$ and, in turn, a two-sided prediction interval.  The IM prediction interval for $\tilde Y$ in this normal prediction problem is the same as the fiducial interval in \citet{wang.hannig.iyer.2012} and the interval in \citet{fertig.mann.1977}.  


\subsection{Log-normal models and an environmental application}
\label{SS:log.normal}

Let $Y=(Y_1,\ldots,Y_n)$ be an iid sample from a log-normal population, with unkown parameter $\theta=(\mu,\sigma)$.  Log-normal models are frequently used in environmental statistics \citep{ott1995}.  In this case, $X=(X_1,\ldots,X_n)$, with $X_i = \log(Y_i)$, will be an iid $\nm(\mu,\sigma^2)$ sample, and the prediction problem can proceed as in Section~\ref{SS:normal} above.  In particular, predicting the next observation $\tilde Y = Y_{n+1}$ is straightforward, so we focus here on something more challenging.  Consider, as in \citet{bhaumik.gibbons.2004}, the problem of finding the upper prediction limit for the arithmetic mean of $m$ future log-normal observations, i.e., $\tilde Y = m^{-1} \sum_{j=1}^m Y_{n+j}$.  Working on the log-scale, with the $X_i$'s, we can first reduce dimension according to sufficiency and then solve for $\theta$ as follows:
\[ \theta(Y,U) = \bigl( \mu(Y,U), \sigma(Y,U) \bigr) = \Bigl( \bar X - \frac{S}{n^{1/2}} \frac{U_1}{U_2}, \, \frac{S}{U_2} \Bigr), \]
where $\bar X = n^{-1} \sum_{i=1}^n X_i$, $S^2 = (n-1)^{-1} \sum_{i=1}^n (X_i-\bar X)^2$, $U_1 \sim \nm(0,1)$, and $(n-1) U_2^2 \sim \chisq(n-1)$, with $U_1,U_2$ being independent.  The marginal association for $\tilde Y$, the arithmetic mean of $m$ future log-normal observations, is 
\[ \tilde Y = \frac1m \sum_{j=1}^m e^{\log Y_{n+j}} = \frac1m \sum_{j=1}^m \exp\Bigl\{ \Bigl( \bar X - \frac{S}{n^{1/2}} \frac{U_1}{U_2} \Bigr) +  \Bigl( \frac{S}{U_2}  \tilde U_{n+j} \Bigr) \Bigr\}, \]
where $\tilde U = (\tilde U_{n+1},\ldots,\tilde U_{n+m})$ are iid $\nm(0,1)$, independent of $U_1$ and $U_2$.  Here we use this association for $\tilde Y$ to construct an upper plausibility prediction limit.  

The above association is not of the separable form in Remark~\ref{re:separable}.  However, for a given $Y$, if $G_Y$ is the distribution of the right-hand side in the previous display, then the marginal association for $\tilde Y$ can be written in the form $\tilde Y = G_Y^{-1}(W)$, for $W \sim \unif(0,1)$, just like in \eqref{eq:pred.marg}.  This completes the A-step.  Since we seek to determine an upper prediction limit, the plausibility function for $\tilde Y$ is given by \eqref{eq:pl.rtside}; see Section~\ref{SS:pc.steps}.

For illustration, we consider an environmental study presented in \citet{bhaumik.gibbons.2004} concerning lead concentration in soil.  It is a ``brownfield'' investigation in which a now-closed plating facility was being investigated for future industrial use.  In April 1996, $m = 5$ soil borings were installed to delineate the extent of lead-impacted soil at the portion of the facility that may have been used for plating. An important environmental question, which \citet{bhaumik.gibbons.2004} addressed using frequentist prediction methods, is to determine whether the on-site mean lead concentration at this area of the facility exceeded background. To facilitate this determination, $n = 15$ off-site soil samples were collected in areas that were uninfluenced by the activities at the facility.  The data are reproduced in Table~\ref{table:bhaumik}. Using the Shapiro-Wilk normality test,  \citet{bhaumik.gibbons.2004} ascertained that, at the 5\% significance level, a log-normal model provides adequate fit to this data. Our main goal in this application is therefore to demonstrate that the on-site concentrations, on average, do not significantly exceed the backgrund. To this end,  we will use the IM framework discussed above to produce an upper prediction limit for the arithmetic mean of lead contents, $\tilde Y$, of $m=5$ on-site soil samples based on the $n=15$ off-site soil samples, $Y$, and then we will compare it to the arithmetic mean of the data collected on the on-site lead concentration.  

\begin{table}
\begin{center}
\begin{tabular}{cccccccccccccccc}
\hline
Off-site & 26 & 63 & 3 & 70 & 16 & 5 & 1 & 57 & 5 & 3 & 24 & 2 & 1 & 48 & 3 \\ 
On-site & 50 & 82 & 95 & 103 & 88 & & & & & & & & & & \\
\hline
\end{tabular}
\end{center}
\caption{Lead (mg/kg) for soil boring samples in off-site and on-site locations.}
\label{table:bhaumik}
\end{table}

The plausibility for $\tilde Y$, for right-sided assertions $A=\{\tilde Y > \tilde y\}$, as a function of $\tilde y$, is shown in Figure~\ref{fig:bhaumik}(a).  Those $\tilde y$ values with plausibility function exceeding 0.05 provide an upper prediction bound for $\tilde Y$ which, in this case, is 136.16 mg/kg.  For comparison, \citet{bhaumik.gibbons.2004} provide the bound 152.26 mg/kg based on their Gram--Charlier approximation, and \citet{kim2007} provides the bound 139.30 mg/kg based on a Bayesian approach.  All three prediction bounds contain the realized arithmetic mean of the on-site data in  Table~\ref{table:bhaumik}, which was 83.6 mg/kg. We therefore conclude that the on-site concentrations do not significantly exceed the backgrund. However, since smaller upper prediction limits are more precise, our IM-based bound is preferred.  An additional advantage of our IM-bound is that it, per Remark~\ref{re:pred.region}, also has a clearer interpretation than the above Bayesian and frequentist bounds.

To check the prediction performance for settings similar to the soil example, we take 5000 samples of size $n=15$ from a log-normal distribution with $\mu=2.173$ and $\sigma^2=2.3808$, the maximum likelihood estimates based on the off-site data in Table~\ref{table:bhaumik}.  A Monte Carlo estimate of the distribution function of $G_Y(\tilde Y)$ is shown in Figure~\ref{fig:bhaumik}(b).  Apparently, $G_Y(\tilde Y)$ is $\unif(0,1)$, so the plausibility function for prediction is valid, by Theorem~\ref{thm:pred.valid}.

For further comparison, we performed a simulation study similar to the one presented in \citet{bhaumik.gibbons.2004}.  We considered three values for $\mu$ (2, 3, 10), six values for $\sigma^2$ (0.0625, 0.2, 0.5, 1, 2, 10), five values for $n$ (5, 10, 20, 30, 100), and three values for $m$ (1, 5, 10).  For each combination, we evaluated the coverage probability of both the lower and upper 90\% prediction intervals.  In all cases, the coverage probability equals the nominal level, up to Monte Carlo error; these estimates are based on 10,000 Monte Carlo samples.  Unlike the Gram--Charlier approximation method  in \citet{bhaumik.gibbons.2004}, our IM-based interval method does not need technical tools for derivation, and achieves the nominal coverage probability even when $\sigma^2 > 3$.  Moreover, the other two frequentist approximation methods reported in \citet{bhaumik.gibbons.2004} do not achieve the nominal coverage probability.   

\begin{figure}
\begin{center}
\subfigure[Plausibility function of $\tilde Y$]{\scalebox{0.55}{\includegraphics{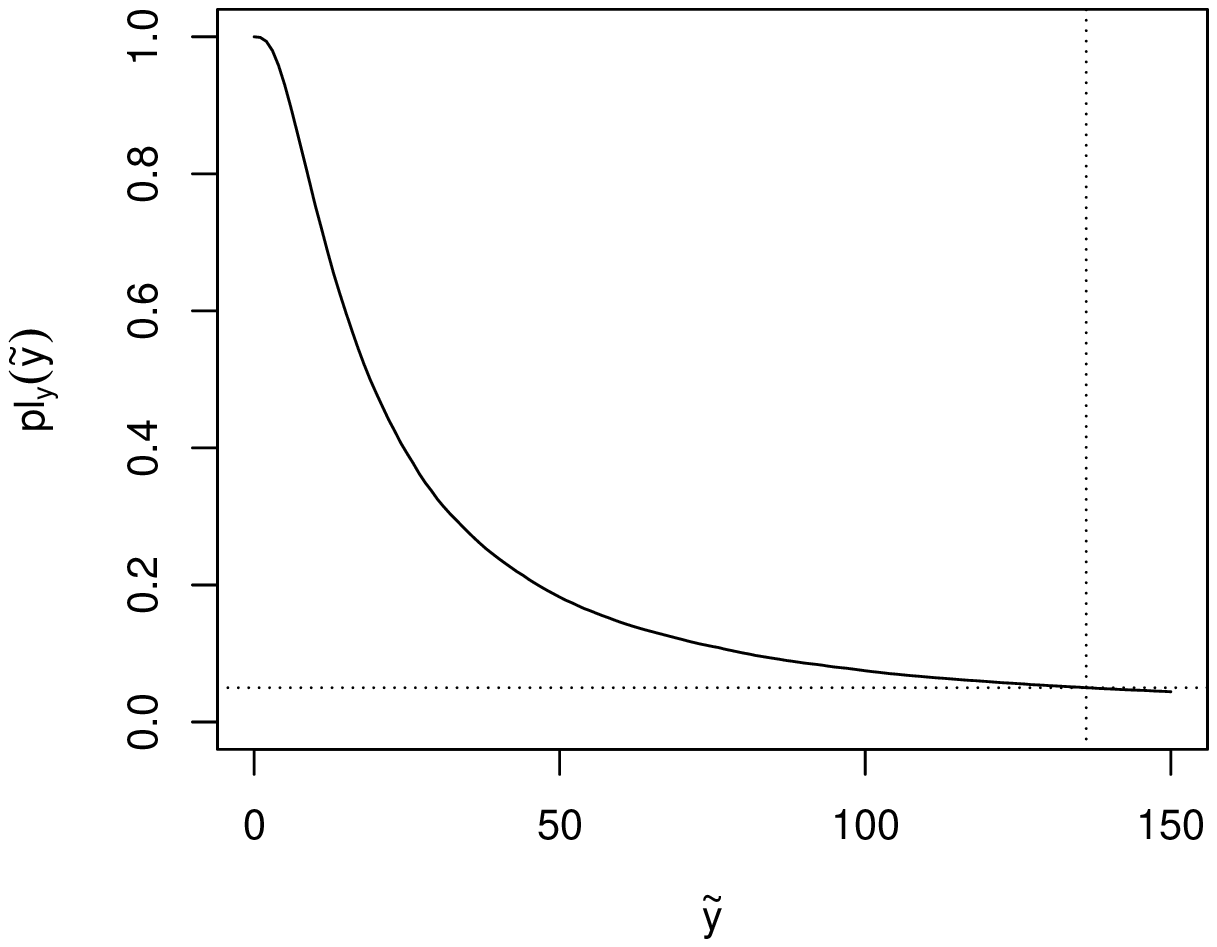}}}
\subfigure[Distribution function of $G_Y(\tilde Y)$]{\scalebox{0.55}{\includegraphics{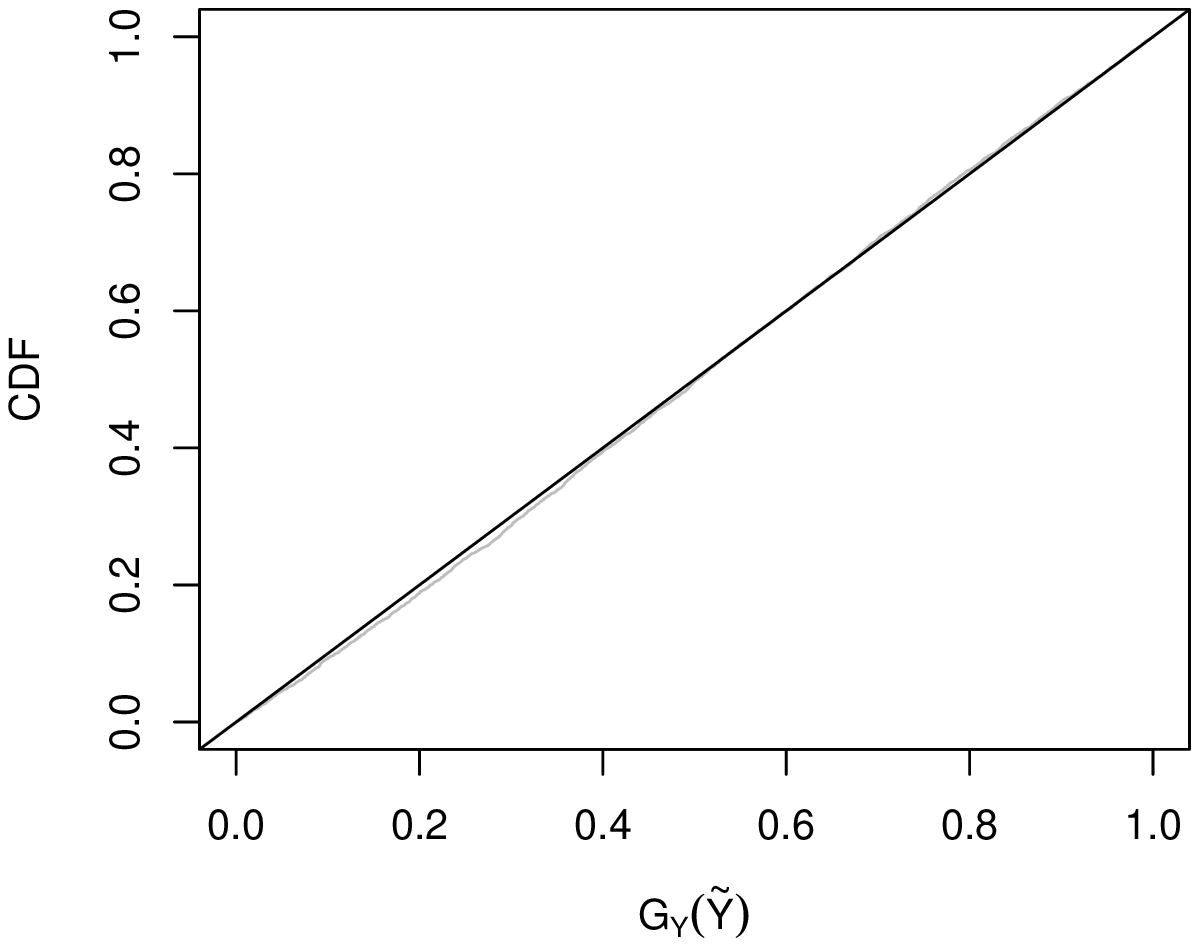}}}
\end{center}
\caption{Panel~(a): Plausibility function of $\tilde Y$ in the log-normal data example.  Panel~(b): Distribution function of $G_Y(\tilde Y)$ (gray) compared with that of $\unif(0,1)$ (black) based on Monte Carlo samples from the log-normal distribution with $\mu=2.173$ and $\sigma^2=2.3808$, the maximum likelihood estimates in the \citet{bhaumik.gibbons.2004} example.}
\label{fig:bhaumik}
\end{figure}

\subsection{Gamma models and a system breakdown application}
\label{SS:gamma}

Let $Y=(Y_1,\ldots,Y_n)$ be an iid sample from a gamma distribution with shape parameter $\theta_1 > 0$ and scale parameter $\theta_2 > 0$, both unknown.  Gamma models are often used in system reliability applications.  Following \citet[][Sec.~5.3]{imcond}, a conditional association for $(\theta_1,\theta_2)$ based on sufficient statistics is given by 
\[ T_1 = \theta_2 \Gamma_{n\theta_1}^{-1}(U_1) \quad \text{and} \quad T_2 = F_{\theta_1}^{-1}(U_2), \]
where $U=(U_1,U_2)$ are iid $\unif(0,1)$, $T_1 = \sum_{i=1}^n Y_i$, $T_2 = n^{-1}\sum_{i=1}^n \log Y_i - \log(n^{-1}T_1)$, $\Gamma_a$ is the ${\sf Gamma}(a,1)$ distribution function, and $F_b$ is a distribution function without a familiar form.  First, suppose the goal is to predict the next (independent) observation $\tilde Y = Y_{n+1}$, with the following association:
\begin{equation}
\label{eq:gamma.assoc1}
\tilde Y = \theta_2 \Gamma_{\theta_1}^{-1}(\tilde U), \quad \tilde U \sim \unif(0,1). 
\end{equation}
Specifically, we want to give a lower prediction limit for $\tilde Y$.  The general strategy is to solve for $\theta=(\theta_1,\theta_2)$ in the conditional association, and then plug this solution in for $\theta$ in the association for the new observation.  In particular, for a given $U=(U_1,U_2)$, write 
\begin{equation}
\label{eq:gamma.soln}
\theta(Y,U) = \bigl( \theta_1(Y,U), \theta_2(Y,U) \bigr) 
\end{equation}
for this solution; it depends on $Y$ only through $(T_1,T_2)$.  The solution exists and is unique, though there is no closed-form expression.  A proof of this claim, along with some details about computing the solution in \eqref{eq:gamma.soln}, are given in the Appendix.  Plugging \eqref{eq:gamma.soln} in to the association for $\tilde Y$ gives the marginal association
\begin{equation}
\label{eq:gamma.assoc2}
\tilde Y = \theta_2(Y,U) \Gamma_{n\theta_1(Y,U)}^{-1}(\tilde U). 
\end{equation}
This association is not of the separable form in Remark~\ref{re:separable}.  In any case, if $G_Y$ denotes the distribution function of the quantity on the right-hand side of \eqref{eq:gamma.assoc2}, then we can write $\tilde Y = G_Y^{-1}(W)$ for $W \sim \unif(0,1)$.  This completes the A-step for IM prediction.  Since we are interested in lower prediction limits, we compute the plausibility function in \eqref{eq:pl.ltside}.    

In some system reliability applications, like in \citet{hamada.etal.2004} and \citet{wang.hannig.iyer.2012}, interest may be in the largest among a collection of $m$ future observations.  In that case, we have an association that looks exactly like \eqref{eq:gamma.assoc2}, except that $\tilde Y = \max\{Y_{n+1},\ldots,Y_{n+m}\}$ is a maximum of $m$ future gamma observations and $\tilde U$ is the maximum of $m$ independent uniforms, independent of $U$.  This involves the same solution $\theta(Y,U)$ as before, so nothing changes except the distributions being used in the Monte Carlo simulation of the plausibility function.  

For illustration, consider the data $Y=(Y_1,\ldots,Y_n)$ on the first breakdown times of $n=20$ machines given in \citet{hamada.etal.2004}.  These data are reproduced in Table~\ref{table:hamada}.  At the 5\% significance level, the Kolmogorov--Smirnov test cannot reject the null hypothesis that these data are gamma, so the goal is to use the IM machinery described above to produce a lower prediction limit for $\tilde Y$, the maximum of $m=5$ future breakdown times.  Proceeding as described above, the plausibility function of $\tilde Y$ is given by $\pl_y(\tilde y) = G_y(\tilde y)$, which can be easily evaluated via Monte Carlo.  A plot of this plausibility function for $A=\{\tilde Y \leq \tilde y\}$, as a function of $\tilde y$, is given in Figure~\ref{fig:hamada}(a).  A one-sided 90\% plausibility interval is the set of all $\tilde y$ values such that $\pl_y(\tilde y) > 0.10$, and the lower bound in this case is 73.53 hours.  For comparison, our lower bound is bigger, i.e., more precise, than the Bayesian lower bound (71.8 hours) in \citet{hamada.etal.2004} and slightly smaller than the fiducial lower bound (74.36 hours) in \citet{wang.hannig.iyer.2012}.  The IM bound, per Remark~\ref{re:pred.region}, also has a clearer interpretation than the Bayesian and fiducial bounds. To assess the performance of the method in problems similar to this one, we simulate 2000 data sets based on the maximum likelihood estimates based on the failure time data.  A Monte Carlo estimate of the distribution function of $G_Y(\tilde Y)$ is shown in Figure~\ref{fig:hamada}(b).  This distribution function is sufficiently close to that of $\unif(0,1)$, so we can conclude our one-sided IM-based 90\% prediction interval has exact coverage.

\begin{table}
\begin{center}
\begin{tabular}{cccccccccc}
\hline
18 & 23 & 29 & 409 & 24 & 74 & 13 & 62 & 46 & 4 \\
57 & 19 & 47 & 13 & 19 & 208 & 119 & 209 & 10 & 188 \\
\hline
\end{tabular}
\end{center}
\caption{Machine first breakdown times, in hours.}
\label{table:hamada}
\end{table}


\begin{figure}
\begin{center}
\subfigure[Plausibility function of $\tilde Y$]{\scalebox{0.55}{\includegraphics{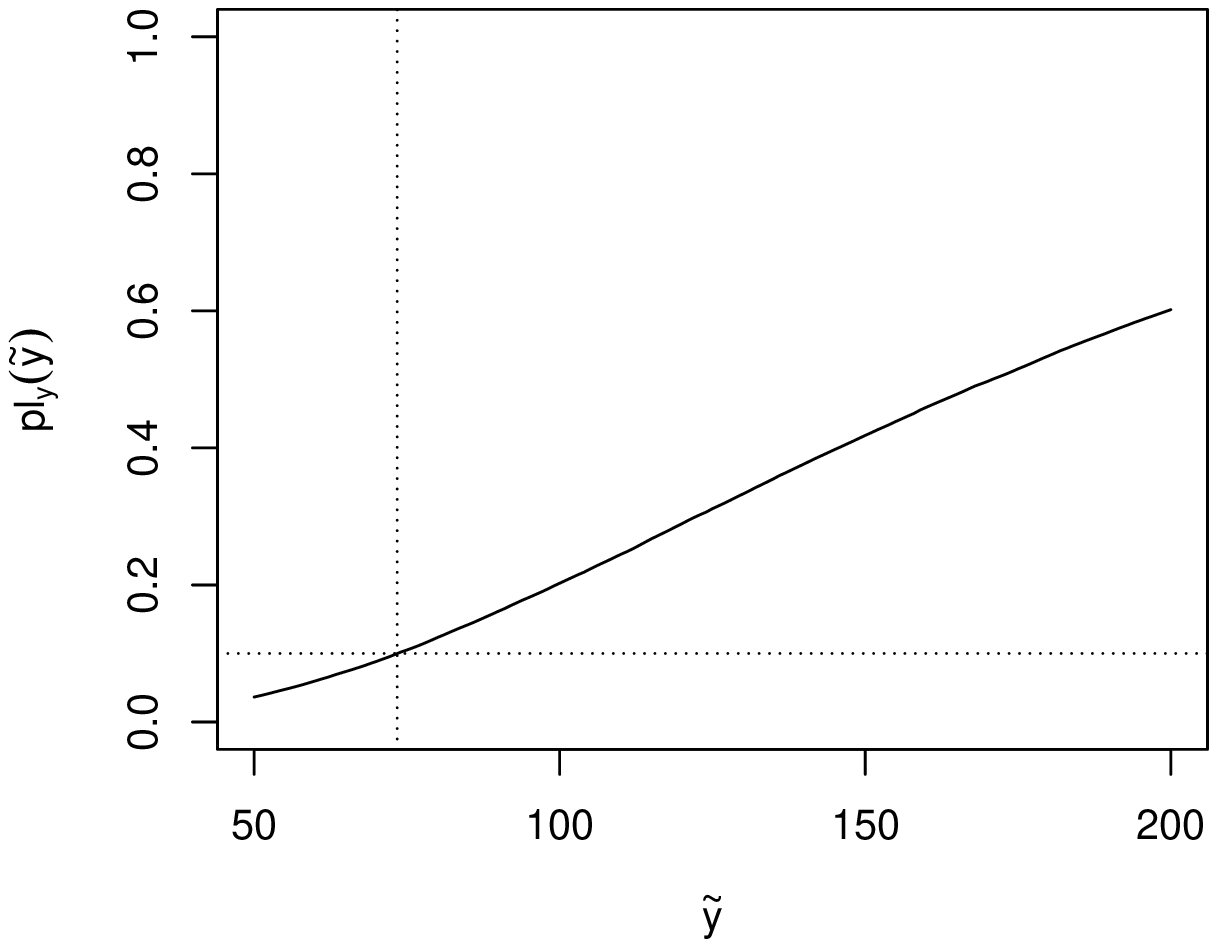}}}
\subfigure[Distribution function of $G_Y(\tilde Y)$]{\scalebox{0.55}{\includegraphics{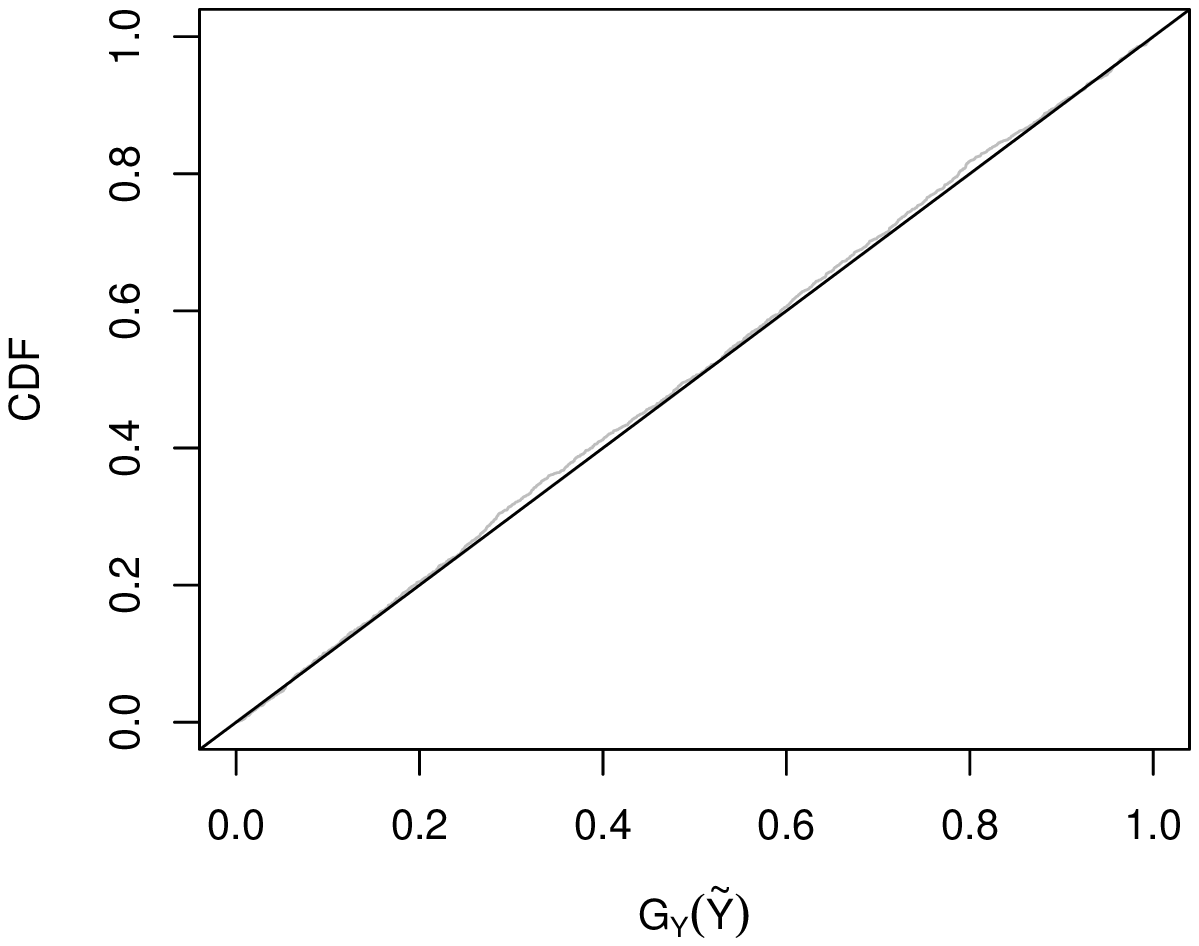}}}
\end{center}
\caption{Panel~(a): Plausibility function of $\tilde Y$ in the gamma data example.  Panel~(b): Distribution function of $G_Y(\tilde Y)$ (gray) compared with that of $\unif(0,1)$ (black) based on Monte Carlo samples from ${\sf Gamma}(\hat\theta_1,\hat\theta_2)$, where $\hat\theta_1=0.8763$ and $\hat\theta_2=90.91$ are the maximum likelihood estimates.}
\label{fig:hamada}
\end{figure}

For further illustration, we consider a simulation experiment, similar to that in \citet{wang.hannig.iyer.2012}, with three values of the sample size $n$ (10, 25, 125), four values of the shape parameter $\theta_1$ (0.5, 1, 5, 10), and two values of $m$ (1, 5); we keep the scale parameter $\theta_2$ fixed at 1.  For each combination, we evaluated the coverage probability of both the lower and upper 90\% prediction intervals based on 10,000 Monte Carlo samples.  In all cases, the coverage probability is within an acceptable range of the target 0.90.

\subsection{Binomial models and a disease count application}
\label{SS:binomial}

Let $Y \sim \bin(n,\theta)$ and   $ \tilde Y \sim \bin(m,\theta)$ be independent binomial random variables, where $n$ and $m$ are known.  The goal is to predict $\tilde Y$ based on observing $Y$.  There is considerable literature on this fundamental problem: see, e.g., \citet{wang2010} for frequentist prediction intervals and \citet{tuyl.etal.2009} for Bayesian prediction intervals.  The starting point for our IM-based analysis is the following joint association for $Y$ and $\tilde Y$, 
\[ F_{n,\theta}(Y-1) \leq 1-U < F_{n,\theta}(Y) \quad \text{and} \quad F_{m,\theta}(\tilde Y-1) \leq 1-\tilde U < F_{m,\theta}(\tilde Y), \]
where $U$ and $\tilde U$ are independent uniforms, and $F_{n,\theta}$ is the $\bin(n,\theta)$ distribution function.  To marginalize over $\theta$, we need a known identity linking the binomial and beta distribution functions, i.e., $F_{n,\theta}(y)  = 1- G_{y+1,n-y}(\theta)$, where $G_{a,b}$ is the $\be(a,b)$ distribution function.  Now rewrite the first expression in the joint association as a $\theta$-interval:
\begin{equation}
\label{eq:binom.assoc.0}
G_{Y, n-Y+1}^{-1}(U) \leq \theta < G_{Y+1, n-Y}^{-1}(U). 
\end{equation}
Next, rewrite the $\tilde Y$ association as 
\[ F_{m,\theta}^{-1}(1 - \tilde U) < \tilde Y < F_{m,\theta}^{-1}(1 - \tilde U). \]
\citet{klenke.mattner.2010} show that $F_{m,\theta}^{-1}(v)$ is an increasing function of $\theta$ for all $v$, so we can ``plug in'' the $Y$-dependent interval for $\theta$ in to this latter inequality, to get 
\begin{equation}
\label{eq:binom.assoc}
F_{m, \theta_1(Y,U)}^{-1}(1 - \tilde U) < \tilde Y < F_{m, \theta_2(Y,U)}^{-1}(1 - \tilde U), 
\end{equation}
where $\theta_1(Y,U)$ and $\theta_2(Y,U)$ are, respectively, the left and right endpoints of the interval in \eqref{eq:binom.assoc.0}.  This completes the A-step.  We are interested in two-sided prediction intervals here; see the P- and C-steps for singleton assertions in Section~\ref{SS:pc.steps}.  Note that this association is an interval, compared to the singletons in the previous examples.  This is a consequence of the discreteness of the binomial, not a limitation of the IM approach; but see below.  Some minor adjustments to the C-step in Section~\ref{SS:pc.steps} is needed to handle this discreteness.  Since we can easily get a Monte Carlo approximation for the distribution of the two endpoints, constructing a plausibility function for $\tilde Y$ is no problem.  

In medical applications, it may be desirable to obtain accurate prediction of the number of future cases of a disease based on the counts in previous years.  \citet[][Sec.~5]{wang2010} gives the following example.  The total number of newborn babies with permanent hearing loss is $Y=23$ out of $n=23061$ normal nursery births over a two-year period.  The goal is to predict $\tilde Y$, the number of newborns with hearing loss in the following year, based on $m=12694$ normal births.  For a two-sided, IM-based 90\% prediction interval for $\tilde Y$, we compute the 5th and 95th percentiles of the distribution of the lower and upper endpoints, respectively, in \eqref{eq:binom.assoc}.  The interval obtained is $(6,21)$, which contains the true $\tilde Y=20$ and is essentially the same as the intervals in \citet{wang2010}; see Remark~\ref{re:pred.region}.  

The plausibility function obtained based on the above construction is a bit conservative.  One possible adjustment, based on an idea presented by \citet{wang.hannig.iyer.2012} in the fiducial context, is to eliminate the interval association for $\tilde Y$ by first eliminating the interval association \eqref{eq:binom.assoc.0} for $\theta$ in terms of the limits $\theta_1(Y,U)$ and $\theta_2(Y,U)$.  The idea is to sample a value, $\hat\theta(Y,U)$, of $\theta$ at random from the interval $(\theta_1(Y,U), \theta_2(Y,U))$.  This results in a modified association for $\tilde Y$:
\begin{equation}
\label{eq:binom.modified}
\tilde Y = F_{m,\hat\theta(Y,U)}^{-1}(1-\tilde U). 
\end{equation}
The intuition is that the uncertainty due to the interval association has been replaced by the uncertainty from sampling.  Since the sampled point is ``less extreme'' than both of the endpoints, this modified association gives a more efficient plausibility function for prediction, which we now demonstrate.  Consider binomial samples of size $n=m=100$ over a range of $\theta$ values.  Here we compare the coverage probability and average lengths of 95\% upper prediction limits based on the modified IM, fiducial, and Jeffreys prior Bayes methods.  We simulated 2500 data sets, and each computation of the prediction interval (modified IM, fiducial, and Bayes) used 10,000 Monte Carlo samples.  In Figure~\ref{fig:binom}, we see that all three methods have coverage slightly above the nominal level over the entire range of $\theta$; this is to be expected, given the discreteness of the binomial model.  The modified IM intervals based on \eqref{eq:binom.modified} tend to have slightly higher coverage probability than the others, but with no perceptible difference in length. 

\begin{figure}
\begin{center}
\subfigure[Coverage probability]{\scalebox{0.55}{\includegraphics{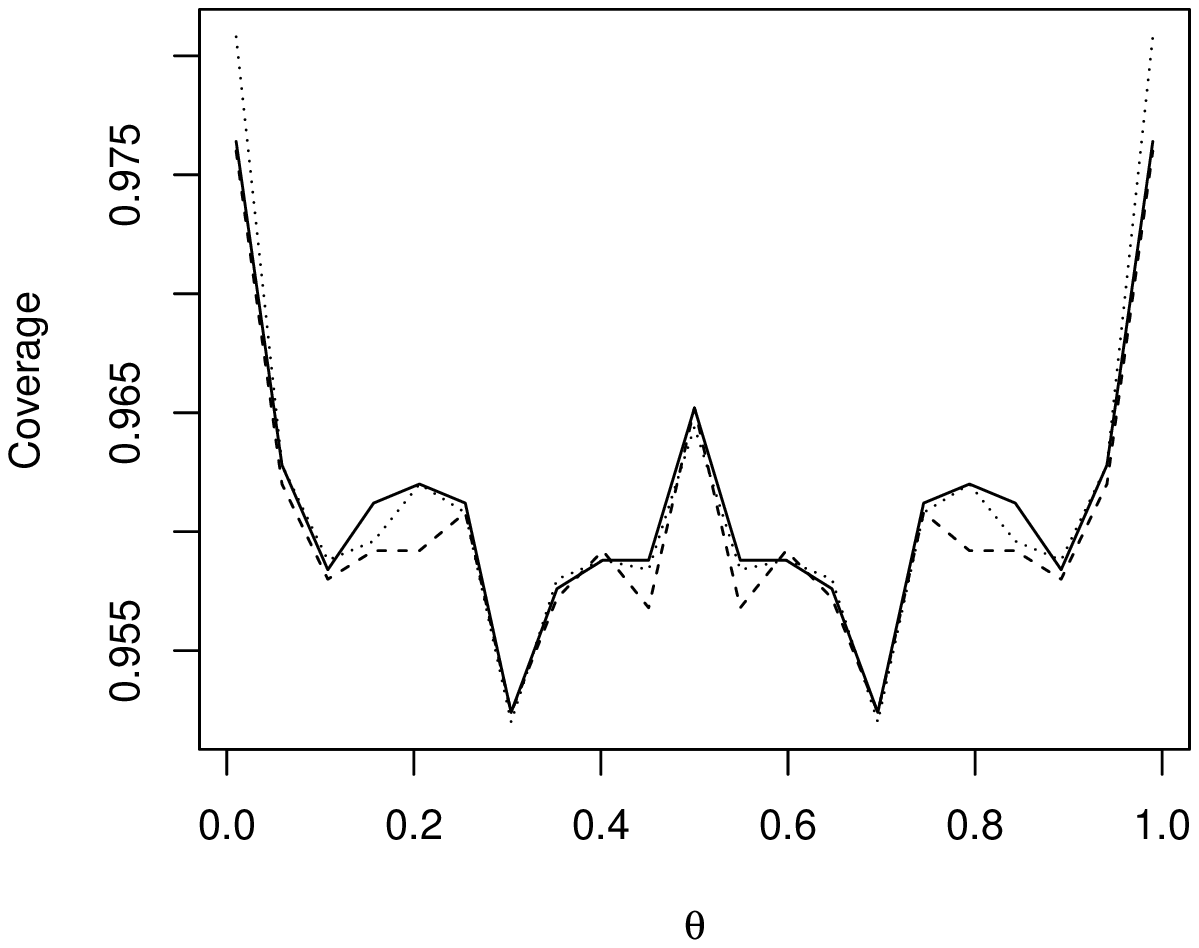}}}
\subfigure[Average length]{\scalebox{0.55}{\includegraphics{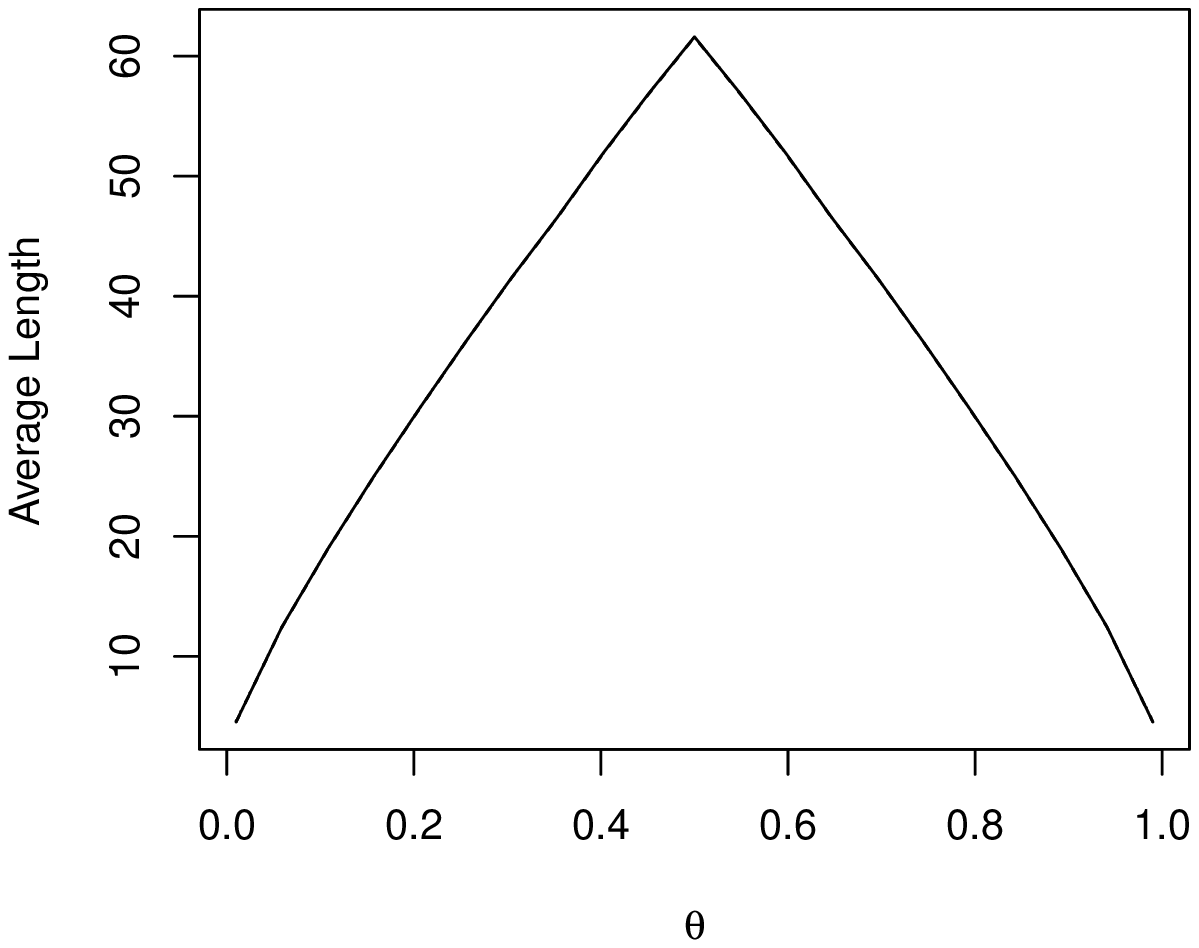}}}
\end{center}
\caption{Coverage probability and average length of the modified IM (solid), fiducial (dashed), and Jeffreys prior Bayes (dotted) upper 95\% prediction intervals, as functions of $\theta$; the three length curves in Panel~(b) are indistinguishable.  Here $n=m=100$ and estimates are based on 2500 simulated data sets.}
\label{fig:binom}
\end{figure}


\section{Some further technical details}
\label{S:further.details}

\subsection{Asymptotic validity}
\label{SS:asymptotics}

Outside the separable class in Remark~\ref{re:separable}, or in cases where $\tilde Y$ is a non-linear function of several future observables, the theory of prediction validity is more challenging.  However, our examples in Section~\ref{S:examples} demonstrate that the uniformity assumption of Theorem~\ref{thm:pred.valid} holds at least approximately.  Here we give a theoretical argument to explain this phenomenon.

Write $Y^n=(Y_1,\ldots,Y_n)$ for data consisting of $n$ iid components.  Suppose that the solution $\theta(Y^n,V)$ converges in probability to $\theta$, as a function of $(Y^n,V)$; this usually is easy to arrange, see the examples  in Section~\ref{S:examples}.  If $\tilde Y$ has a continuous distribution, then, without loss of generality, we can write $G_{Y^n}(\tilde Y) = \tilde F_{\theta(Y^n,V)}(\tilde Y)$, where $\tilde F_\theta$ is the true distribution of $\tilde Y$.  Trivially, we have $(\theta(Y^n,V),\tilde Y) \to (\theta, \tilde Y)$ in distribution so, if $(\theta, \tilde y) \mapsto \tilde F_\theta(\tilde y)$ is continuous, then the continuous mapping theorem implies that 
\[ G_{Y^n}(\tilde Y) = \tilde F_{\theta(Y^n,V)}(\tilde Y) \to \tilde F_\theta(\tilde Y) \sim \unif(0,1) \quad \text{in distribution}. \]
Therefore, we can generally be sure that the distribution of $G_{Y^n}(\tilde Y)$ will be approximately $\unif(0,1)$ when the sample size $n$ is large.  This argument holds even if $\tilde Y$ is some scalar function of several future observations. 

In addition to providing an asymptotic validity result, the argument above is also relevant to prediction accuracy.  That is, we have demonstrated that the IM ``predictive distribution'' for $\tilde Y$,  which mixes $\tilde F_{\theta(Y^n,V)}$ over the distribution of $V$, converges to $\tilde F_\theta$, the true distribution of $\tilde Y$.  A precise prediction accuracy result requires computing a measure of the distance/divergence of the IM predictive from the truth.  We expect that results comparable to those in \citet{lawless.fredette.2005} can be derived, but we leave this as a question to be considered in future work.

\subsection{Case of dependent $\tau(U)$ and $\eta(U)$}
\label{SS:conditional}

Recall that, in Section~\ref{SS:setup}, it was assumed that the original association, $Y=a(\theta,U)$, could be decomposed as $T(Y) = b(\theta, \tau(U))$ and $H(Y) = \eta(U)$, and, furthermore, that $\tau(U)$ and $\eta(U)$ are independent.  Of concern here is the case where $\tau(U)$ and $\eta(U)$ are not independent, which would arise, say, in models that are not regular exponential families.  The IM framework is equipped to handle this, but the details are more complicated.  Here, we describe the three-step construction of an IM for prediction in the case where the model is a location shift of a Student-t distribution.  

Let $Y=(Y_1,\ldots,Y_n)^\top$ be iid, with $Y_i=\theta + U_i$ and $U_i \sim \stt_d$, i.e., $\theta$ is a location parameter and the error has a Student-t distribution with known degrees of freedom $d$.  The case with an additional unknown scale parameter can be handled similarly.  Let $T(Y)$ be the maximum likelihood estimator for $\theta$, and let $H(Y)=Y - T(Y) 1_n$ be the vector of residuals, where $1_n$ is a (column) $n$-vector of unity.  Since $T$ is equivariant in this example, we have 
\[ T(Y) = \theta + T(U) \quad \text{and} \quad H(Y) = H(U), \]
so that $\tau=T$ and $\eta=H$.  However, $\tau(U)$ and $\eta(U)$ are not independent here.  As suggested in \citet{imcond} and in Section~\ref{SS:setup} above, we want to consider the conditional distribution of $\tau(U)$, given that $\eta(U)$ equals the observed value of $H(Y)$.    Let $H_0$ denote the observed value of $H(Y)$.  Then \citet{imcond} give a formula for the conditional distribution of $\tau(U)$, given $\eta(U)=H_0$.  Write $\prob_{V|H_0}$ for this conditional distribution of $V=\tau(U)$, so that 
\[ T(Y) = \theta + V, \quad V \sim \prob_{V|H_0}. \]
This can be solved for $\theta$, and a marginal association connecting the observed $Y$, the next $\tilde Y$ to be predicted, and the pair $(V,\tilde U)$ of auxiliary variables is of the form  
\[ \tilde Y = T(Y) + \tilde U - V, \quad \text{where} \quad V \sim \prob_{V|H_0}, \; \tilde U \sim \stt_d. \]
Computation with the conditional distributions is more cumbersome, but our claim that the dependent case is conceptually no different than the independent case should now be clear.  Moreover, using the law of iterated expectation, it can be shown that the conditioning does not affect validity result in Theorem~\ref{thm:pred.valid}.

\subsection{Multivariate prediction}
\label{SS:multi}

The focus of this paper was on the case of predicting a scalar $\tilde Y$, which is possibly a scalar-valued function of several future observables.  However, in some cases there could be interest in simultaneous prediction of several future observables.  One example is in regression, where interest may be in predicting the response value corresponding to several values of the predictor variables simultaneously.  While the general IM framework is well-equipped to handle the multivariate case, our developments here have employed a few scalar-specific steps.  Our goal in this section is simply to highlight those scalar-specific steps, defining a roadmap to extend the present developments to the multivariate case.  

First, note that, up to the simplified association formula in \eqref{eq:pred.marg}, there is nothing in the developments in Section~\ref{SS:setup} specific to the scalar $\tilde Y$ case.  There, we introduced a distribution $G_Y$ and a probability integral transform, which is only valid in the scalar case.  However, \eqref{eq:pred.marg0} is well-defined for vector $\tilde Y$, and would conclude the A-step for multivariate prediction.  From here, the P-step proceeds by introducing a predictive random set for the pair $(V,\tilde U)$.  Validity, as usual, would not be a major obstacle, but an efficient choice of predictive random set would be problem specific.  For a scalar auxiliary variable, the only reasonable choice of predictive random set is an interval, but in the multivariate case, there are lots of ``reasonable'' shapes, and the choice among them makes a difference in terms of the corresponding IM's efficiency.  Work on the construction of efficient predictive random sets, in general, is ongoing, and results to be obtained will have immediate application to the multivariate prediction problem.  

Second, observe that the discussion of one-sided assertions in Section~\ref{SS:pc.steps} is not appropriate in the multivariate setting, where there is no proper ordering.  The two most natural assertions would be the singletons and assertions defined via level sets of some scalar-valued function of $\tilde Y$, e.g., balls $\{\tilde y: \|\tilde y\| \leq r\}$ for some fixed radius $r > 0$.  The latter case reduces to the scalar prediction problem covered in this paper.  For the singleton assertion case, nothing in the present development needs to change, except that the predictive random set, in general, must be specified for the pair $(V,\tilde U)$ directly.  If this predictive random set is valid, then the conclusion of Theorem~\ref{thm:pred.valid} holds.  Again, the challenge is that the shape of the predictive random set is directly related to the shape and efficiency of, say, the IM prediction regions for $\tilde Y$.  So, more work on constructing good/optimal predictive random sets for multivariate auxiliary variables is needed.

\section{Conclusion}
\label{S:discuss}

In this paper, we have proposed a method for prediction of future observables based on the recently developed IM framework.  The key to the IM approach in general is the association of data and parameters with unobservable auxiliary variables, and the use of random sets on the auxiliary variable space to construct belief and plausibility functions on the parameter space.  In the context of prediction of future observables, all the model parameters are nuisance, and an extreme form of the marginalization technique described in \citet{immarg} is required, which allows us to reduce the dimension of the auxiliary variables, increasing efficiency.  We give conditions which guarantee that the IM for prediction is valid, and we argue that this notion of IM validity translates to frequentist coverage guarantees for our plausibility intervals for future observables.  A sequence of practical examples demonstrates the quality performance of the proposed method, along with its generality and overall simplicity.  

The methodology described here covers both discrete and dependent-data problems.  However, these problems present unique challenges.  For example, in the binomial example in Section~\ref{SS:binomial}, our standard IM approach was valid but conservative.  A modified and more efficient IM was proposed, and its validity was confirmed numerically, but a theoretical basis for this modification is required.  For dependent data problems, marginalization to reduce the dimension of auxiliary variables as described here is possible, but the details would be more challenging.  Moreover, as discussed in Section~\ref{SS:multi}, additional work is needed to properly extend the developments in the present paper to the multivariate prediction problem.  These are all topics for future research.  

To conclude, recall the take-away message from Section~\ref{S:intro}.  The IM approach provides a general and easy-to-implement method for constructing valid prior-free probabilistic summaries of the information in observed data relevant for inference or prediction.  The fact that these summaries can be converted to frequentist procedures with fixed-$n$ performance guarantees and comparable efficiency compared to existing methods is an added bonus.  We expect further developments and applications of IMs in years to come.

\section*{Acknowledgments}

The authors thank Chuanhai Liu for helpful comments, Nicholas Karonis for cluster computer access, and John Winans for computational assistance.  This research is partially supported by the U.S.~National Science Foundation, DMS--1208833.

\appendix

\section{Technical details for Section~\ref{SS:gamma}}
\label{S:appendix.gamma}

\subsection{Existence and uniqueness of the solution \eqref{eq:gamma.soln}}
\label{SS:appendix.existence}

Here the issue is existence and uniqueness of the solution $\theta(T,U) = (\theta_1(T,U), \theta_2(T,U))$ in the gamma problem in Section~\ref{SS:gamma}.  The only non-trivial part is the solution $\theta_1$ of equation $F_{\theta_1}(t_2)=u_2$, involving only $(t_2,u_2)$.  The challenge is that $F_{\theta_1}$ is a non-standard distribution.  \citet{glaser1976}, in his notation, considers the random variable 
\[ U^\star = \Bigl\{ \frac{(\prod_{i=1}^n Y_i)^{1/n}}{\frac1n \sum_{i=1}^n Y_i} \Bigr\}^n, \]
the $n$-th power of the ratio of geometric and arithmetic means of an iid $\gam(\theta_1,1)$ sample.  Then $T_2=n^{-1}\log(U^\star)$, i.e., our $T_2$ is a monotone increasing function of Glaser's $U^\star$.  A consequence of Glaser's Corollary~2.2 is that $U^\star$ is stochastically strictly increasing in $\theta_1$, which implies that $F_{\theta_1}(t_2)$ is a decreasing function of $\theta_1$ for all $t_2$.  Therefore, if a solution exists for $\theta_1$ in \eqref{eq:gamma.soln}, it must be unique by monotonicity.  

Turning to the existence of a solution for $\theta_1$, we need to show that, for any $t_2$, $F_{\theta_1}(t_2)$ spans all of the interval $(0,1)$ for $u_2$ as $\theta_1$ varies.  By monotonicity, it suffices to consider the limits $\theta_1 \to \{0,\infty\}$.  \citet{jensen1986} considers the random variable $W=-1/T_2$ and shows, in his Equation~(9), that $\theta_1/W$ has a limiting distribution as $\theta_1 \to \{0,\infty\}$, which implies the same for $\theta_1 T_2$.  It is now clear that $F_{\theta_1}(t_2)$ converges to 1 and 0 as $\theta_1$ converges to $0$ and $\infty$, respectively, for all $t_2$.  Therefore, a solution for $\theta_1$ in \eqref{eq:gamma.soln} exists for all $(t_2,u_2)$ pairs, as was to be shown.

\subsection{Computing the solution \eqref{eq:gamma.soln}}
\label{SS:appendix.computing}

Here we consider computing the solution $\theta(Y,U)$ in \eqref{eq:gamma.soln}.  The only challenging part is solving for $\theta_1$, so we shall focus on this.  Suppose $t_2$ and $u_2$ are given, and define a function $r(x) = F_x(t_2)-u_2$; the goal is to find the root for $r$.  One can evaluate $r(x)$ by simulating $\gam(x,1)$ variables, giving a Monte Carlo approximation of $F_x(t_2)$, and the root can then be found with any standard method, e.g., bisection.  However, this can be fairly expensive computationally.  A more efficient alternative approach is available based on large-sample theory.  By Theorem~5.2 in \citet{glaser1976} and the delta theorem, if $n$ is large, then $F_x$ can be well approximated by a normal distribution function with mean $\psi(x) - \log(x)$ and variance $n^{-1}\{\psi'(x) - 1/x\}$, where $\psi$ and $\psi'$ are the digamma and trigamma functions, respectively.  With this normal approximation, it is easy to evaluate $r(x)$ and find the root numerically.  Though this is based on a large-sample approximation, in our experience, there is no significant loss of accuracy, even for small $n$.  

The normal approximation discussed above is simply a tool to find the solution $\theta(Y,U)$.  It also provides some intuition related to the asymptotic argument in Section~5.1.  When $n$ is large, the variance in the normal approximation is $O(n^{-1})$, so the distribution function $F_x$ will have a steep slope in the neighborhood of the solution to the equation $t_2=\psi(x)-\log(x)$ and, therefore, the root for $r(x)$ will be in that same neighborhood, no matter the value of $u_2$.  The solution to the equation $t_2=\psi(x)-\log(x)$ is the maximum likelihood estimator of $\theta_1$ \citep[e.g.,][]{fraser.reid.wong.1997}, which is consistent.  Therefore, when $n$ is large, \eqref{eq:gamma.assoc2} and \eqref{eq:gamma.assoc1} are essentially the same, so the approximate validity of the corresponding prediction plausibility function is clear.  

One last modification that we found to be helpful was to modify that normal approximation discussed above by replacing the normal distribution function with a gamma.  That is, find solutions for the mean and variance of the normal approximation as before, but then use a gamma distribution function with mean and variance matching those obtained for the normal.  See the R code available at the first author's website.

\bibliographystyle{apalike}
\bibliography{/Users/rgmartin/Dropbox/Research/mybib}

\end{document}